\documentclass[pra,twocolumn,showpacs,superscriptaddress]{revtex4-2} 
\pdfoutput=1

\usepackage{graphicx}

\usepackage{physics}
\usepackage{bbold}
\usepackage{amsmath, amsthm, amssymb,mathrsfs}
\newtheorem{prop}{Proposition}

\usepackage{mathtools}

\usepackage[normalem]{ulem}
\usepackage{optidef}
\usepackage{algorithmic}
\usepackage{algorithm}

\usepackage{comment}
\usepackage{xcolor}
\usepackage{vwcol}

\usepackage[acronym]{glossaries}
\makeglossaries
\newacronym{qkd}{QKD}{quantum key distribution}
\newacronym{di}{DI}{device-independent}
\newacronym{pm}{PM}{prepare-and-measure}
\newacronym{sdp}{SDP}{semidefinite programming}
\newacronym{POVM}{POVM}{Positive Operator Valued Measure}
\newacronym{usd}{USD}{Unambiguous state discrimination}
\newacronym{qber}{QBER}{Quantum Bit Error Rate}
\newcommand{\id}{\mathbb{1}}

\usepackage{multirow}
\usepackage{diagbox}

\usepackage[normalem]{ulem}

\usepackage[colorlinks]{hyperref}

\begin{document}

\title{Receiver-Device-Independent Quantum Key Distribution Protocols}

\author{Marie Ioannou}
\affiliation{Department of Applied Physics University of Geneva, 1211 Geneva, Switzerland}
\author{Pavel Sekatski}
\affiliation{Department of Applied Physics University of Geneva, 1211 Geneva, Switzerland}
\author{Alastair A.\ Abbott}
\affiliation{Department of Applied Physics University of Geneva, 1211 Geneva, Switzerland}
\affiliation{Univ.\ Grenoble Alpes, Inria, 38000 Grenoble, France}
\author{Denis Rosset}
\affiliation{Department of Applied Physics University of Geneva, 1211 Geneva, Switzerland}
\author{Jean-Daniel Bancal}
\affiliation{Department of Applied Physics University of Geneva, 1211 Geneva, Switzerland}
\affiliation{Université Paris-Saclay, CEA, CNRS, Institut de physique théorique, 91191, Gif-sur-Yvette, France}
\author{Nicolas Brunner}
\affiliation{Department of Applied Physics University of Geneva, 1211 Geneva, Switzerland}

\begin{abstract}
    We discuss quantum key distribution protocols and their security analysis, considering a receiver-device-independent (RDI) model. The sender's (Alice's) device is partially characterized, in the sense that we assume bounds on the overlaps of the prepared quantum states. The receiver's (Bob's) device requires no characterisation and can be represented as a black-box. Our protocols are therefore robust to any attack on Bob, such as blinding attacks. In particular, we show that a secret key can be established even when the quantum channel has arbitrarily low transmission by considering RDI protocols exploiting sufficiently many states. Finally, we discuss how the hypothesis of bounded overlaps can be naturally applied to practical devices.
\end{abstract}

\maketitle

\section{Introduction}

Quantum key distribution (QKD) \cite{Bennett1984,Ekert1991} allows two users to establish a secret key via a quantum channel and an authenticated but public classical channel. QKD, together with the one time pad method, provides a secure method of communication with information-theoretical security~\cite{Vernam1926}. Indeed, unlike classical schemes, the security of QKD protocols is physical: it only relies on some knowledge about the functioning of the devices controlled by the communicating parties and the general laws of quantum mechanics.  Nevertheless, different approaches require different levels of detail in how the devices are modeled \cite{Scarani2009,Lo2014,Diamanti2016,Xu2020}. The ``standard'' approach presumes a full description of different elements in the setup. Such QKD systems are available commercially and can reach high rates over long distances.

However, relying on a detailed quantum model for characterizing the devices may open backdoors that quantum hackers can exploit. Indeed a mathematical model always represents (at best) an idealization of a practical device. For example, the well-known ``blinding attacks'' exploit the fact that standard models for describing photon detectors typically fail when the intensity of the incoming light falls outside their working range \cite{Lydersen2010,Gerhardt2011}. When a fair-sampling type assumption is used on top of this, the door is open to attacks where an eavesdropper Eve obtains full information about the key, without introducing any detectable level of errors (fair sampling assumes that the occurrence of no-detection events is independent of the choice of the measurement setting \cite{Pearle,berryFS,OrsucciFS}).

This motivates the investigation of the stronger, device-independent (DI) approach. Here, devices are viewed as classically controlled black boxes, and the security of QKD protocols can be demonstrated \cite{Acin2007,Pironio2009,Vazirani,Rotem2018} assuming only that (i) the devices can be described accurately within quantum mechanics, and (ii) no information about the secret key leaks out of the laboratories of Alice and Bob (the two communicating parties). While this approach represents, in principle, the perfect solution to counter any hacking attack, its practical implementation is highly challenging, requiring the distribution of high-quality entanglement and notably high detection efficiencies (the best current protocol demands $68.5\%$ \cite{liu2021}). First proof-of-principle experiments have recently been reported \cite{nadlinger2021,zhang2021,liu2021}, but any practical implementation of DI QKD is arguably still far out of reach.

Beyond the standard (device-dependent) approach and the DI one, there exists a broad range of models that can be considered, where some of the devices are fully (or partially) characterized, while others are treated as black boxes. The most well-known is arguably the measurement-DI (MDI) approach \cite{Lo2012,Braunstein12}, which has been extensively studied and realized experimentally achieving record distances (see, e.g., \cite{Yin2016,Pittaluga2021,Chen2021511,Chen2021658}).

In parallel, another approach has been investigated, considering an asymmetric scenario where one of the end parties is trusted, while the other one is fully untrusted. Referred to as ``one-sided DI'', this model was first proposed in Ref.~\cite{Tomamichel2011}. Establishing a connection to quantum steering, Ref.~\cite{Branciard2012} then investigated the practical limitations of such a protocol, in particular the resilience to noise and losses. Unfortunately, an implementation turns out to be challenging, as the requirements in terms of detection efficiencies ($>65.9 \%$) are only slightly relaxed compared to the full DI model. Other works \cite{Tomamichel2012,Tomamichel2017}, following up on Ref.~\cite{Tomamichel2011}, discussed the implementation in a prepare-and-measure scenario. While considering the effect of noise and finite-size data, these works do however not take into account the effect of losses. Instead, a fair-sampling type assumption is made, which opens the door to blinding attacks, as in standard protocols; see, e.g., \cite{Acin2016}. Hence these results cannot be applied to a practical QKD setup (where losses are unavoidable) without sacrificing the one-sided DI security. Finally, another approach, termed semi-DI \cite{Pawlowski2011,WP2015,Goh2016}, considered a prepare-and-measure scenario assuming an upper bound on the dimension of the prepared quantum systems. Again, these protocols are unpractical, requiring detection efficiencies comparable to the full DI model \cite{dallarno2015}.

In this work we present QKD protocols that achieve one-sided DI security and that are amenable to a practical prepare-and-measure implementation. We refer to these protocols as being ``receiver-device-independent'' (RDI). A specific example of such a protocol was recently presented, along with an experimental realisation, in the companion paper~\cite{Ioannou2021}.
Here, we present a more general class of RDI-QKD protocols and provide a detailed theoretical analysis, investigating the possibilities and limits of QKD in RDI scenarios.

We thereby consider a prepare-and-measure scenario, where the sender (Alice) uses a partially characterized device, while the receiver (Bob) uses an untrusted device. The protocol being black-box on Bob's side, it is therefore inherently secure against attacks on the receiver, notably blinding attacks \cite{Lydersen2010,Gerhardt2011}. On Alice's side, the characterisation we require consists in providing bounds on the (complex) overlaps of the prepared states (given formally by a Gram matrix).
We moreover discuss how this hypothesis can be naturally applied to practical devices.

In practice, the RDI scenario can be quite naturally motivated. Consider for instance a large company communicating with an end-user. The latter has essentially no means to test their cryptographic device, which is therefore conveniently treated as a black-box. On the other hand, the company has access to advanced technology and technical expertise, and can therefore regularly test and characterize their cryptographic device. We note that the MDI approach is not applicable to this scenario, as both Alice and Bob require a trusted device (while trust is then relaxed on an intermediate relay station).

The paper is organized as follows. 
In Section~\ref{sec:scenario} we present the scenario of RDI QKD and discuss the key assumptions that are made, before outlining the RDI-QKD themselves in Section~\ref{sec:protocols}. 
In Section~\ref{sec:securityAnalysis} we present a detailed security analysis. In the noiseless case we present an analytical security proof, showing that our protocols can achieve the maximal distance possible in an RDI scenario. Specifically, we show that it is possible to obtain a positive key rate for any transmission $\eta>1/n$, where $n$ denotes the states prepared by Alice, and corresponds also to the number of measurements performed by Bob. Our protocols can therefore accomodate any amount of losses in principle (by considering sufficiently many states), and are optimal in terms of robustness to losses, as no secret key can be obtained when $\eta\leq1/n$ \cite{Acin2016}. When noise is present, the security analysis relies on semidefinite programming, for which we adapt the method introduced in Ref. \cite{wang_characterising_2019}, providing lower bounds on the key rate. Then, in Section~\ref{sec: noisy states}, we discuss the practical relevance of our RDI approach, in particular how bounds on the overlaps (Gram matrix) can be estimated and justified in practice. Finally, in Section~\ref{sec:comparison} we discuss how our protocol compares to other QKD protocols and scenarios. 

\section{Scenario}
\label{sec:scenario}

\begin{figure}[b!]
    \centering
    \includegraphics[width=0.4\textwidth]{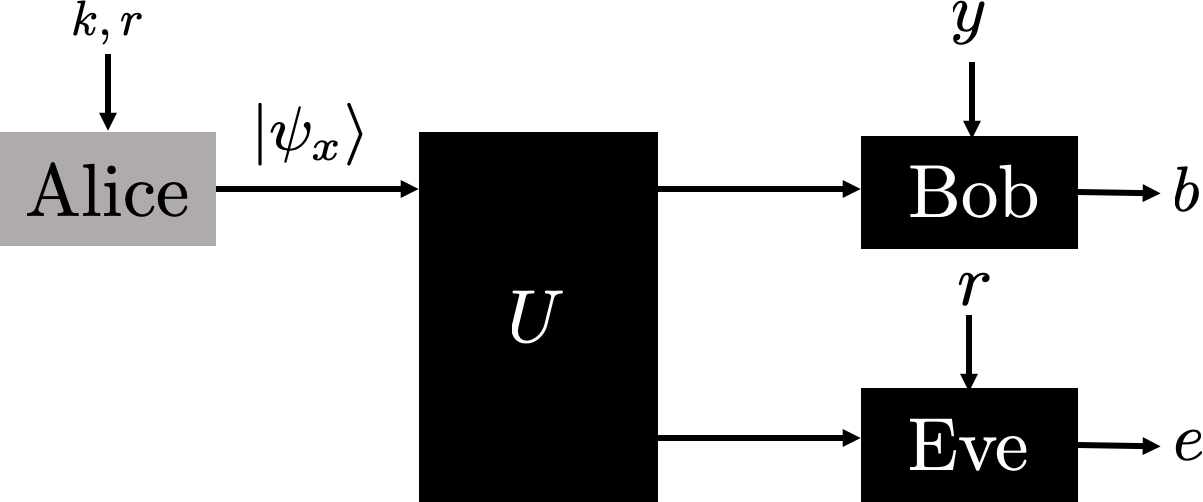}
    \caption{Scenario: Alice and Bob can establish secret key based on the Gram matrix $G$ of the set of states $\{\ket{\psi_x}\}_x$ prepared by Alice and the observed data $p(b|x,y)$. Eve has a complete control on the quantum channel, and can also have full knowledge of the functioning of the devices of Alice and Bob.}
    \label{fig:scheme}
\end{figure}

We consider a prepare-and-measure scenario as shown in Fig.~\ref{fig:scheme}. Alice sends, over a public quantum channel, one state out of a set of $n$ states $\{\ket{\psi_x}\}_{x=0}^{n-1}$. Bob chooses among $n$ measurements labelled by $y=0,\dots,n-1$. All measurements have binary outputs $b=0,1$. After many rounds, Alice and Bob can estimate the probability distribution $p(b|x,y)$. Bob's measurement device is completely uncharacterized and can be seen as a black box with an input $y$ and an output $b$. The black box feature is a requirement if we aim to design a protocol robust to attacks where Eve controls Bob's device. The key assumption we make on the setup is about Alice's preparations. Namely, we assume that all inner-products $ \gamma_{ij} = \braket{\psi_i}{\psi_j} $ are bounded. These assumptions do not fix the total dimension of the Hilbert space and only partially characterize Alice's device.

The assumption that Alice prepares pure states with known inner-products $\gamma_{ij}$ simplifies the presentation and analysis of the protocol, but is evidently impossible to fulfil exactly in practice.
In Sec.~\ref{sec: noisy states} we revisit this assumption on Alice's preparation device and show how the presence of noise, unavoidable in experiments, can also be analyzed within our framework in several ways. In particular, we show that the general situation where the preparation device is subject to fluctuating noise, which remains within a certain parameter window, can be analyzed by taking inequality constraints on (the real and imaginary parts of) the values $\gamma_{ij}$.

Besides the assumption on Alice's preparation device, specific to our protocol, we also make the standard QKD assumptions, also made in the DI scenario: (i) Alice's input $x$ and Bob's measurement setting $y$ are completely uncorrelated from Eve; (ii) Eve only has access to the classical and quantum communication specified by the protocol, she cannot gather any additional information about $x$ and $y$; (iii) We assume the validity of quantum physics. In the following, Eve is restricted to collective attacks. She interacts with each round independently and can store her system in a quantum memory. 

As we will see, a lower bound on the raw secret key rate, can be computed solely from the observed statistics $p(b|x,y)$, given that the setup satisfies the assumptions detailed above. For some ideal cases (no noise), we derive analytical bounds. More generally, e.g. in the presence of noise, we obtain bounds via semi-definite programming (SDP) adapting the methods introduced in Ref. \cite{wang_characterising_2019}.

\section{Protocols}
\label{sec:protocols}

In this section we describe the general structure of the RDI-QKD protocols we consider and give a family of concrete examples. 

\subsection{General structure}

We begin by presenting the general structure of our RDI-QKD protocols. 

Consider a given ensemble of states $\{\ket{\psi_x}\}_{x=0}^{n-1}$ that Alice is able to prepare and binary measurements $\{B_{0|y},B_{1|y}\}_{y=0}^{n-1}$ that Bob can perform. We can now define protocols with a general structure as follows, where the steps 1 and 2 are repeated sufficiently many times in order to guarantee a final key of desired length.

\begin{algorithm}[H]
\floatname{algorithm}{RDI-protocol}
\renewcommand{\thealgorithm}{}
\caption{Steps to generate a secret key between Alice and Bob.}
\label{protocol1}
Alice and Bob share an authenticated classical channel as well as a quantum channel. Steps 1 and 2 are repeated sufficiently many times, before proceeding to steps 3 and 4.
    \begin{enumerate}
        \item \textbf{Raw key generation}
            \begin{algorithmic}[1]
                \STATE Alice randomly chooses a pair of integers ${\mathbf{r}}=(r_0,r_1)$ with $0\leq r_0 <r_1\leq n-1$ and a bit $k=0,1$. According to her choice she sends the state $\ket{\psi_{x=r_k}}$ over the quantum channel to Bob.
                \STATE Bob randomly chooses an integer $y$ with $0\leq y\leq n-1$ and performs the binary measurement $\{B_{0|y},B_{1|y}\}$ on the state received from Alice. 
            \end{algorithmic}
        \item \textbf{Sifting} \\Alice and Bob use the classical channel to communicate.
            \begin{algorithmic}[1]
                \IF{$b=1$}
                    \STATE{Bob tells Alice to discard the round.}
                \ELSIF{$b=0$}
                    \STATE{Bob asks Alice to reveal $\mathbf{r}$.}
                    \STATE{Alice reveals $\mathbf{r}$.}
                        \IF{$y=r_0$ \OR $y=r_1$}
                            \STATE{Bob tells Alice the round is conclusive.}
                        \ELSE
                            \STATE{Bob tells Alice to discard the round.}
                        \ENDIF    
                \ENDIF
            \end{algorithmic}
           
             \item \textbf{Parameter estimation}
            \item \textbf{Error correction and privacy amplification}
    \end{enumerate}
\end{algorithm} 

This structure defines a broad class of protocols specified by the choices of $n$, the states $\{\ket{\psi_x}\}_{x=0}^{n-1}$, and the measurements $\{B_{0|y},B_{1|y}\}_{y=0}^{n-1}$.
In general, the idea is to choose states and measurements such that, in Step 2.7, Bob can readily infer from the observed outcome $b$ what the key bit $k$ of Alice is. 
Below we will describe in more detail some specific examples, which will clarify the principles behind the RDI protocols we describe.

Here, we are not going to describe the classical steps 3 and 4 in detail, as under the assumption of collective attacks these steps can be performed with standard techniques. In  step $3$ Alice and Bob reveal their registers for a subset of rounds chosen at random, allowing them to estimate the probability $p(b|x,y)$.  In step 4, Alice and Bob perform standard one-way error correction followed by privacy amplification protocols, enabling them to extract the final secret key from the raw key available after step 2. Detail of such protocols can be found in the reviews of Refs.~\cite{Scarani2009,Xu2020}. For the security analysis presented in the next section we thus focus on the raw key, under the assumption of collective attacks and a known probability distribution $p(b|x,y)$. The security analysis under coherent attacks is left for future work.

\subsection{Ideal qubit protocol}
\label{subseq:concreteex}

We describe a class of protocols based on qubit states and measurements. As we will see later, these protocols can be considered ideal in the sense of being optimal from the point of view of robustness to loss. At this point, however, we present the protocol in the case of no loss and no noise.

Alice prepares states from a set of $n$ single-qubit states $\{\ket{\psi_x}\}_{x=0}^{n-1}$ with
\begin{equation}\label{eq:qubitStates}
\ket{\psi_x}=\cos(\theta/2)\ket{0}+e^{\frac{i 2\pi}{n} x}\sin(\theta/2)\ket{1}
\end{equation}
for some given $\theta$.
Following the general protocol outlined above, to encode the raw key bit Alice chooses a pair of integers $\mathbf{r}=(r_0,r_1)$ with $0\leq r_0 < r_1 \leq n-1$, among ${n \choose 2}$ possible pairs. For a key bit $k$, Alice sets $x=r_k$. Note that every state $x$ can encode the bit value $0$ or $1$. Alice sends $\ket{\psi_{x=r_k}}$ via the quantum channel to Bob. Bob has $y=0,\dots,n-1$ measurements and each measurement has a binary output $b=0,1$. The output $b=1$ corresponds to a projection onto $\ket{\psi_y}$ while $b=0$ corresponds to the projection on the orthogonal subspace $\id -\ketbra{\psi_y}{\psi_y}$. If Bob observes $b=0$, he can with certainty exclude the state $x=y$. We refer to the rounds where $b=0$ as conclusive rounds. If the round is conclusive, Bob asks Alice to reveal $\mathbf{r}$. If $y=r_0$ or $y=r_1$, Bob is able to infer the raw key bit and announces to Alice that the round is successful; otherwise he tells Alice to discard the round.

The security analysis of this protocol in the presence of noise and loss, is described below in Section~\ref{sec:qubitSecAnal}. 
Moreover, in Section~\ref{sec:analyticBounds} we show that this protocol is optimal within RDI-QKD protocols in the sense that it yields a positive key rate for any $\eta>\frac{1}{n}$, arbitrarily close to the threshold of $1/n$ beyond which no secret key can be established \cite{Acin2016}.

\subsection{Towards practical protocols}

While the above ideal qubit protocol is useful to test the limits of model, the RDI approach can also be used quite naturally, and give good protocols, in more realistic setups. 

Firstly, the requirement that Alice prepares pure states is not necessary.
Indeed, the case of mixed states can naturally be encompassed by considering purifications of the states Alice prepares.
We discuss how to take into account the overlap assumption on Alice's device in this case in Section~\ref{subseq:mixedstates}.

Secondly, the qubit protocol described above can be adapted quite naturally to an optical setup, where a dimension bound on the states Alice prepares is unrealistic. This is because only the overlaps of the prepared states is required (their Gram matrix), but not their Hilbert space dimension. One can therefore consider a protocol where polarized coherent states of light are prepared, as reported recently in the companion paper~\cite{Ioannou2021}. Therein a proof-of-principle implementation of such a  protocol was reported, achieving finite-size key over a 4.8km optical fiber.

\section{Security analysis} 
\label{sec:securityAnalysis}

Eve's information about the secret bit $k$ is bounded by assuming that the Gram matrix $G$ of the set of encoding states is fully characterized and that the probabilities $p(b|x,y)$ are perfectly estimated by Alice and Bob. The Gram matrix $G$ is a Hermitian matrix whose entries are given by 
\begin{equation}
    G_{ij} = \braket{\psi_i}{\psi_j}.
\end{equation}
We do not bound the dimension of the Hilbert space associated to the system sent by Alice. However, under the assumption that Alice prepares pure states the rank of the Gram matrix equals the dimension of the subspace spanned by these states. Recall that this assumption is not indispensable for our analysis, and will be relaxed in Sec.~\ref{sec: noisy states}. Furthermore, no characterization of the exact encoding, transmission channel nor measurement device is needed. Eve can correlate herself to the states prepared by Alice, she can design Bob's measurement device by the means of an ancilla and a unitary operation, and she can use a quantum memory to keep her ancilla until the end of the classical post-processing (cf.\ Fig.~\ref{fig:scheme}). In fact, she can keep her ancilla until any later time and wait until the reconciliation between Alice and Bob is over in order to perform a measurement allowing her to extract as much information as possible about the secret bit $k$. 

The asymptotic key rate (per round) is lower bounded by \cite{Devetak2005} 
\begin{equation}\label{eq: keyrate 1}
\left[H(k|\text{Eve, succ}) - H(k| \text{Bob, succ})\right]p(\text{succ}),
\end{equation}
where $H(k|\text{Eve(Bob), succ})$ is the entropy of $k$ conditional on Eve(Bob) and the fact that a round is not discarded, and $p$(succ) is the probability that a round is not discarded. Bob's entropy can be upper-bounded as $H(k| \text{Bob},\text{succ})\leq h_2(\text{QBER})$, where $h_2(\cdot)$ is the binary entropy and QBER is the quantum bit error rate. Eve's conditional entropy can be lower-bounded by the conditional min entropy 
\begin{equation}\nonumber \begin{split}
    H(k|\text{Eve},\text{succ}) &\geq H_{\min}(k|\text{Eve},\text{succ}) \\& =-\log_2\left(p_g(e=k|\text{succ})\right),
\end{split}
\end{equation}
which is in a one-to-one relation with the maximal probability $p_g(e=k|\text{succ})$ that Eve guesses the bit $k$ correctly~\cite{konig2009operational} if the round was not discarded. Combing the two arguments, we can lower bound the key rate by the quantity
\begin{equation}\label{eq: keyrate}
R = \left[-\log_2\left(p_g(e=k|\text{succ})\right) - h_2(\text{QBER})\right]p(\text{succ}).
\end{equation}

The QBER and $p(\text{succ})$ are extracted from the observed statistics $p(b|x,y)$ while the guessing probability $p_g(e=x|\text{succ})$ needs to be upper bounded in order to give a lower bound on $R$. Note that $p(\text{succ})>0$: if $p(\text{succ})=0$ there is no raw key generation and hence nothing for Eve to guess. The guessing probability is given by
\begin{widetext}
    \begin{equation}
        \begin{aligned}
            p_g(e=k|\text{succ}) &= \frac{p(e=k,\text{succ})}{p(\text{succ})} \\
            &=\frac{\sum_{r=0}^{{n \choose 2}-1} p_R(r) \sum_{k=0}^1 p_K(k) \sum_{y=0}^{n-1 } p_Y(y) \trace{(\rho_{r_k}^{BE}M_{1|y} E_{k|r}) (\delta_{y,r_0}+\delta_{y,r_1})}}{\sum_{r=0}^{{n \choose 2}-1} p_R(r)\sum_{k=0}^1 p_K(k) \sum_{y=0}^{n-1} p_Y(y) \trace{(\rho_{r_k}^{BE}M_{1|y}\mathbb{1}) (\delta_{y,r_0}+\delta_{y,r_1})}},
        \label{eq:p_g}
        \end{aligned}
    \end{equation}
\end{widetext}
where $M_{b|y}$ are Bob's measurement operators with $b=0,1$ and $y=0,\dots,n-1$, and $E_{k|r}$ are Eve's measurement operators with $k=0,1$ and $r=0,\dots,{n\choose 2}-1$. $p_R(r)$, $p_Y(y)$ and $p_K(k)$ are the probabilities of choosing the inputs $r$, $y$ and $k$. Hence, $\sum_r p_R(r)=\sum_kp_K(k)=\sum_yp_Y(y) = 1$, $p_K(k)\geq0$ $\forall k$, $p_Y(y)\geq0$ $\forall y$ and $p_R(r)\geq0$ $\forall r$. Here we will always we take the input probabilities to be uniformly random over all inputs. As already mentioned, the dimension of the problem is not bounded, so without loss of generality we can, using Naimark's dilation theorem, assume that Bob's and Eve's measurements are projectors satisfying the following properties:
\begin{equation}
    \begin{aligned}
        &M_{b|y}M_{b'|y} = \delta_{b,b'}M_{b|y} \;\; &\forall y \\
        &\sum_b M_{b|y} = \mathbb{1} \;\;&\forall y \\
        &E_{e|\mu}E_{e'|\mu} = \delta_{e,e'}E_{e|\mu} \;\; &\forall \mu \\
        &\sum_e E_{e|\mu} = \mathbb{1} \;\; &\forall \mu\\
        &[M_{b|y},E_{e|\mu}]=0 \;\; &\forall b,e,y,\mu.
    \end{aligned}
    \label{eq: propop}
\end{equation}
The last property comes from the fact that Bob and Eve act on two different Hilbert spaces. Note that we do not perform any fair-sampling type assumption on Bob's measurement. The cases were no clicks are recorded at Bob will be included in one of the outputs $b$; see Section IV.B.

\subsection{Semidefinite programming approach}

Since $p(\text{succ})$ is extracted from the observed statistics, to upper bound $p_g(e=k|\text{succ})$ we need just to upper bound $p(e=k,\text{succ})$.
To do this, we will use the method presented in \cite{wang_characterising_2019}. 
In particular, we use the approach described therein which provides a semidefinite programming (SDP) hierarchy giving increasingly tight outer approximations of the set of quantum correlations in discrete prepare-and-measure scenarios compatible with a given Gram matrix. The hierarchy is known to converge to the actual set of quantum correlations, whereas for a fixed level it provides a tractable method of bounding the guessing probability over correlations compatible with the observed statistics.
This problem would, without the hierarchy, be computationally intractable since no bound on the Hilbert space dimension is assumed.

Let $\{S_i\}_{i=0}^{s-1}$ be a set of measurement operators and define the moment matrix $\Gamma$ of size $n s\times n s$ as
\begin{equation}
    \Gamma = \sum_{x,x'=0}^{n-1} \Gamma_{xx'} \otimes \ketbra{\hat{e}_x}{\hat{e}_{x'}},
\end{equation}
where $\{\ket{\hat{e}_x}\}_{x=0}^{n-1}$ is an orthonormal basis of $\mathbb{R}^n$ and we recall that $n$ is the number of states prepared by Alice. The sub-blocks $\Gamma_{xx'}$ are defined as 
\begin{equation}
    \Gamma_{xx'} = \sum_{i,j=0}^{s-1} \bra{\psi_x}S_i^\dagger S_j \ket{\psi_x'} \otimes \ketbra*{\hat{\hat{e}}_j}{\hat{\hat{e}}_j}
\end{equation}\\
where $\{\ket*{\hat{\hat{e}}_i}\}_{i=0}^{s-1}$ is an orthonormal basis of $\mathbb{R}^s$.
It is easily shown that the moment matrix $\Gamma$ is positive semidefinite.
The elements of the set $\{S_i\}_{i=0}^{s-1}$ are monomials of the operators $B_{b|y}$ and $E_{e|\mu}$. This set of operators can be chosen arbitrarily but the aim is to have as many linearly independent operators as possible in the moment matrix.
By taking all monomials of measurement operators up to a given order, we can define different levels of the hierarchy. The first two levels are given, e.g., by the two following sets of operators:
 \begin{equation}
     \begin{aligned}
         \mathcal{S}_1 &= \{\mathbb{1},B_{b|y},E_{e|\mu}\},\\
         \mathcal{S}_2 &=  \mathcal{S}_1 \cup \{B_{b|y}B_{b'|y'},E_{e|\mu}E_{e'|\mu'},B_{b|y}E_{e|\mu}\},
     \end{aligned}
 \end{equation} 
and the levels $\mathcal{S}_n$ for $n> 2$ can likewise be defined inductively. Ref.~\cite{wang_characterising_2019} proved that as $n$ goes to infinity (i.e., in the infinite level limit), the hierarchy converges to the set of quantum correlations.

For the sake of clarity, we define $\Gamma_{xx'}^{ST}:=\bra{\psi_x}S^\dagger T \ket{\psi_{x'}}$ with $S,T\in \mathcal{S}$ and $x,x'=0,...,n-1$. The SDP upper bounding $p(e=x,\text{succ})$ is given by
\begin{subequations}
    \begin{align}
        \max_{\Gamma} \;\;&{\frac{1}{(n-1)n^2}\sum_{r=0}^{n \choose 2}\sum_{k=0}^1\sum_{y=0}^{n-1} \Gamma_{r_kr_k}^{B_{0|y}E_{r_k|r}}(\delta_{y,r_0}+\delta_{y,r_1})}\label{eq: sdpobj}\\
        \text{s.t.} \;\;\;&\Gamma_{xx'}^{\mathbb{1}\mathbb{1}}= \braket{\psi_x}{\psi_{x'}} = \gamma_{xx'} \;\;\;\;\;\; \forall  x,x' \label{eq: sdpconsG}\\
        &\Gamma_{xx}^{\mathbb{1}B_{b|y}} = p(b|x,y)  \;\;\;\;\;\;\;\;\;\;\;\;\;\;\forall b,x,y \label{eq: sdpconsp}\\
        &\trace(\Gamma_{xx'} F_{k}) = f_k \;\;\;\;\;\;\;\;\;\;\; k =0,\dots,m, \forall x,x' \label{eq: sdpconsMiguel}\\
        &\Gamma \succeq 0. & \label{eq: sdpconspos}
    \end{align}
\end{subequations}
The overlap constraint between the set of states is enforced by Eq.~\eqref{eq: sdpconsG}. Eq.~\eqref{eq: sdpconsp} enforces the moment matrix $\Gamma$ to be compatible with the observed correlations $p(b|x,y)$. In Eq.~\eqref{eq: sdpconsMiguel} $F_k$ are hermitian matrices and $f_k$ complex coefficients which are defined in order to encode the constraints on Bob's and Eve's operators given by Eq.~\eqref{eq: propop}, as well as the constraints between elements of $\Gamma_{xx'}$ implied by the fact that $\Gamma_{xx'}^{ST}=\Gamma_{xx'}^{S'T'}$ whenever $S^\dagger T = {S'}^\dagger T'$ (cf.\ Prop.~4 of Ref.~\cite{Navascues2008}). 

\subsection{Security analysis of the ideal qubit protocol}
\label{sec:qubitSecAnal}

Here, we will analyze the security of the idealized qubit protocol presented in Section \ref{subseq:concreteex}, including in the presence of loss and noise. We will model noise by the means of a depolarizing channel with parameter $\lambda \in [0,1]$, which replaces the transmitted state with a maximally mixed state with probability $\lambda$ \cite{NC11}. Loss is modeled by a binary erasure channel \cite{NC11} with erasure probability $(1-\eta)$, $\eta\in[0,1]$. Such a model of loss assumes that loss is orthogonal with respect to the encoding, which is typically the case if one considers, e.g., the polarization of photons for the encoding of the secret bit. 

The Gram matrix $G$ corresponding to the set of states \eqref{eq:qubitStates} prepared by Alice is given by
\begin{equation}
    G_{ij} =\cos^2(\theta/2)+e^{i\frac{2\pi(i-j)}{n}}\sin^2(\theta/2)
\label{eq: gramn}
\end{equation}
with $i,j=1,...,n$.
The probability distribution is then given by
\begin{equation}
\begin{split}
    p(b=0|x,y) &= \eta \left( \frac{\lambda}{2}+ (1-\lambda)\sin^2(\theta)\, \sin^2\left(\frac{\pi(x-y)}{n} \right) \right).
\end{split}
\end{equation}

Given the Gram matrix $G$ of \eqref{eq: gramn} and the observed probability distribution one can upper bound the secret key rate as shown previously. Figure~\ref{fig:simulations} shows the raw key rate as a function of the transmission $\eta$ for different QBER's and values of $n$. For each $\eta$ we numerically optimized over $\theta$ to obtain the optimal $R$. We notice that the lower-bound on the key rate goes asymptotically to zero as $\eta \rightarrow 1/n$. This is optimal because at $\eta = 1/n$, Eve can break the security by intercepting the states sent by Alice and forcing Bob's detector according to her outcome and Bob's input (see Section~\ref{sec:analyticBounds}). Therefore, for any prepare-and-measure protocol, the key rate is null for $\eta\leq 1/n$.

Interestingly, B92~\cite{Bennett1992} is a special case of the proposed protocol with $n=2$ and a fixed $\theta = \frac{\pi}{4}$. Under the same assumptions, our protocol outperforms B92 with respect to the transmission and the noise tolerance, see Fig.~\ref{fig:simulationsUsVsWorld}. Also, BB84~\cite{Bennett2014} under the same assumptions is outrun by our protocol with 3 states.

\begin{figure}
    \centering
    \includegraphics[width=0.4\textwidth]{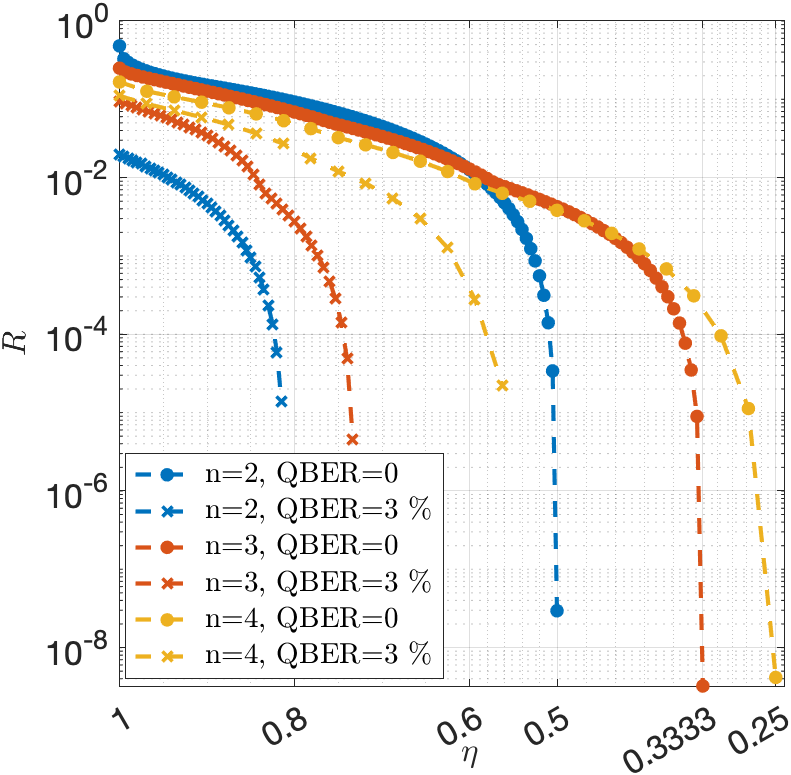}
    \caption{Raw key rate for our RDI-QKD protocol. The graph shows the lower bound on the raw key rate $R$ as a function of the transmission for different number of states and QBER's. For $n$ states, the noiseless protocol has a positive key rate down to $\eta=1/n$, which is the minimal transmission for which this is possible in any prepare-and-measure scenario. The protocol is also tolerant to noise in state preparation.}
    \label{fig:simulations}
\end{figure}

\begin{figure}
    \centering
    \includegraphics[width=0.4\textwidth]{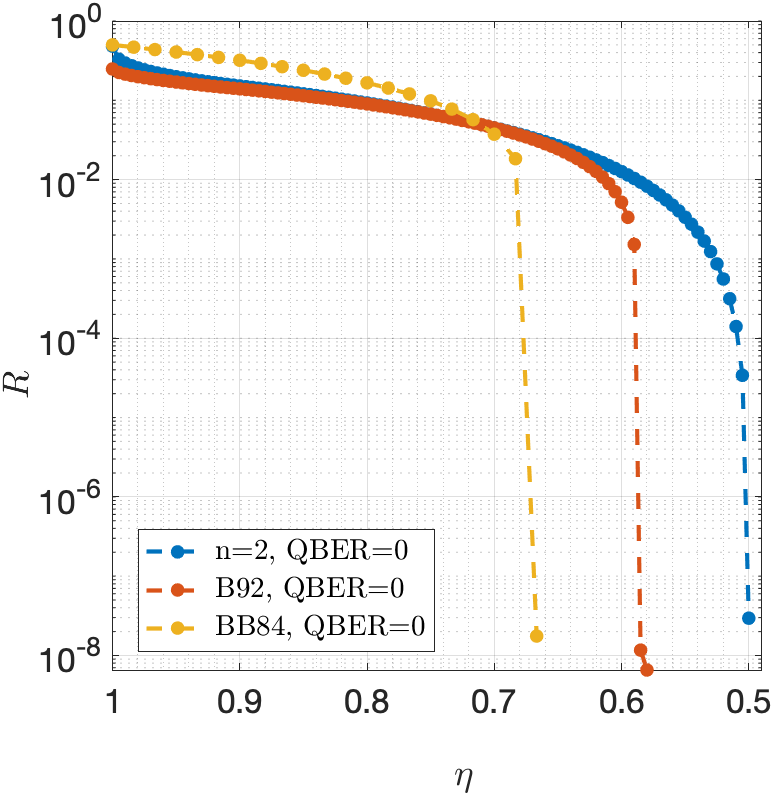}
    \caption{Comparison of our RDI-QKD protocol with other protocols under the same assumptions. The RDI protocol with $n=2$ outperforms BB84 and B92.}
    \label{fig:simulationsUsVsWorld}
\end{figure}

\subsection{Analytical bounds}
\label{sec:analyticBounds}

In this section we prove analytically that, if Alice prepares sufficiently many states, the protocol can in principle tolerate arbitrary small transmission $\eta$. First, with an explicit attack from Eve we lower bound the transmission $\eta$ required to have $R>0$. (Proposition~\ref{prop: ub}). Secondly, we show that this bound is tight as long as $G$ is chosen to obey an additional natural condition (Proposition~\ref{prop: lb}). That is, for any transmission $\eta$ exceeding the threshold, Eve is unable to guess the secret bit with certainty in all rounds, giving rise to a positive key rate.

Transmission loss in the line (scaling with distance) and finite detection efficiency are the bottlenecks in most QKD protocols. Both effects give rise to a loss channel and contribute to the total transmission $\eta$. In this section we assume that this loss is the only imperfection in the setup. This captures the main limiting factor of real QKD setups and allows us to derive relatively simple analytical bounds. 
We assume that loss is orthogonal with respect to the secret bit encoding, such that with probability $\eta$ the system sent by Alice is lost and Bob observes a third outcome (e.g., a no-click event $b=\emptyset$). Bob then attributes it the value $b=1$, such that the rounds where the system sent by Alice is lost are rejected in the protocol.

In this case any protocol with a Gram matrix $G_{ij}=\braket{\psi_i}{\psi_j}$ with $i,j = 0,\dots,n-1$ and the honest measurements ${B}_{1|y} =\ketbra{\psi_y}{\psi_y}$ with ${B}_{0|y}=\mathbb{1}- B_{1|y}$ leads to measurement probabilities
\begin{equation} \label{eq: probas loss only}
\begin{split}
p(b=0|x,y)&= \eta(1-|G_{xy}|^2),\\
p(b=1|x,y)&=1- p(b=0|x,y),
\end{split}
\end{equation}
with $x,y = 0,\dots,n-1$. One notes that with such probabilities $p(0|x=y)=0$: Bob's bits are perfectly correlated to Alice's after the sifting, i.e.\ $h_2(\textrm{QBER})=0$. For the following, we define $\lambda_\text{min}(G)$ as the minimal non-zero eigenvalue of the Gram matrix $G$.
\begin{prop}
\label{prop: ub}
    Given a Gram matrix $G\in\mathbb{C}^{n\times n}$ and measurement probabilities of Eq.~\eqref{eq: probas loss only}, a necessary condition for $R>0$ is that $\eta > \frac{1}{n-\lambda_\mathrm{min}(G)}$.
\end{prop}

\begin{proof}
    Let us assume that with probability $q$ Eve intercepts the state sent by Alice and makes an unambiguous state exclusion measurement $M_i = \mu(\mathbb{1}-\ketbra{\psi_i}{\psi_i})$ with $i=0,\dots,n-1$, $\mu\in[0,1]$ 
    and $M_{n} = \mathbb{1} - \sum_{i=0}^{n-1} M_i$. 
    
    If Eve obtains an outcome $i<n$, she can exclude with certainty the state $\ket{\psi_i}$, whereas if she gets the outcome $n$ she cannot conclude anything. In order to have as many conclusive outcomes as possible Eve maximizes $\mu$ under the constraint $M_{n}\geq0$: 
    \begin{equation}
        \begin{aligned}
            \max_\mu \;\;\;&\mu\\
            \text{s.t.}\;\;\; &\mathbb{1}\frac{(n\mu-1)}{\mu} \leq \sum_{i=0}^{n-1} \ketbra{\psi_i}{\psi_i},\\
            &\mu \geq 0.
        \end{aligned}
        \label{eq: maxmu}
    \end{equation}   
    The first constraint in Eq.~\eqref{eq: maxmu} is satisfied if the eigenvalues of $\sum_{i=0}^{n-1} \ketbra{\psi_i}{\psi_i}$ are all larger than $\frac{(n\mu-1)}{\mu}$. 
    But the eigenvalues of $\sum_i \ketbra{\psi_i}{\psi_i}$ coincide with the nonzero eigenvalues of the Gram matrix $G$. Hence, the above maximization is satisfied if $\frac{(n\mu-1)}{\mu}\leq \lambda_\text{min}(G)$. This leads to an optimal $\mu^*=\frac{1}{n-\lambda_\text{min}(G)}$ and $p(i|x)= \mu^*(1-|G_{xi}|^2)$. The result $i$ of Eve's measurement is then sent to Bob's detector which only outputs $b=0$ if $y=i$, i.e.\ $p(b=0|y,i)=\delta_{y,i}$. The resulting probability observed by Bob is 
    \begin{equation}\begin{split}
    p(b=0|x,y) &=\sum_{i=0}^n p(b=0|i,y) p(i|x) \\
    & = \mu^*(1-|G_{xy}|^2).
    \end{split}
    \end{equation}
    With probability $(1-q)$ Eve does not intercept the message, and Bob's detector is instructed to perform the ideal measurement $p(b=0|x,y)=(1-|G_{xy}|^2)$. Eve wants to remain undetected and hence needs to reproduce the expected statistics of Eq.~\eqref{eq: probas loss only}.  Her attack must thus satisfy the equality
    \begin{equation}
    \label{eq: probAnalBound}
        \eta(1-|G_{xy}|^2)=q\mu^*\big(1-|G_{xy}|^2) + (1-q)(1-|G_{xy}|^2)\big)
    \end{equation}
    for all $x,y$. This implies that Eve can not intercept the message more often than in a fraction $q=\frac{1-\eta}{1-\mu^*}$ of rounds. In particular, if $q=\frac{1-\eta}{1-\mu^*}\geq 1$ or $\eta \leq \frac{1}{n-\lambda_\text{min}(G)}$ she can intercept the message in every round resulting in $p(y=i|\text{succ})=p_g(e=k|\text{succ})=1$ and $R=0$. 
    \end{proof}
    
    More generally, this attack gives a lower bound on Eve's guessing probability as
    \begin{equation}\begin{split}
        p_g(e=k|\text{succ}) &\geq q + (1-q)\frac{1}{2} \\
        &= \frac{1}{2}\left(1+\frac{1-\eta}{\eta(n-(1+\lambda_\text{min}(G))}\right),
    \end{split}
    \end{equation}
    with equality if $p_g(e=k|\text{succ})=\frac{1}{2}$ for the honest implementation at $\eta=1$.

For the considered family of protocols the proposed attack allows Eve to guess the secret bit $k$ of Alice perfectly whenever one has $\eta\leq \frac{1}{n-\lambda_\text{min}(G)}$. The converse question is whether, for any transmission exceeding this value, there exists a protocol (with a given $n$ and $\lambda_\text{min}(G)$) yielding a strictly positive key rate. We will now show that this is indeed the case by considering a qubit protocol with rank$(G)=2$, as discussed in Sec.~\ref{subseq:concreteex}.

\begin{prop}
\label{prop: lb}
Consider a Gram matrix $G$, with $\rank(G)=2$,
leading to measurement probabilities in Eq.~\eqref{eq: probas loss only}. 
If the transmission exceeds  $\eta >\frac{1}{n-\lambda_\mathrm{min}(G)}$, then one can obtain a positive key rate, i.e., $R>0$.
\end{prop}
\begin{proof}
Since in our case $h_2(\textrm{QBER})=0$, from Eq.~\eqref{eq: keyrate} one sees that the condition $R>0$ is equivalent to $p_g(e=k|\text{succ})<1$, that is Eve can not always guess the secret bit with certainty. Thus, we want to prove $p_g(e=k|\text{succ})<1$. To do so we will proceed by assuming $p_g(e=k|\text{succ})=1$ and reach a contradiction. 

To start, it is convenient to replace our prepare-and-measure scenario by an equivalent entanglement-based scenario. Alice prepares an entangled state
\begin{equation}
\ket{\Phi}_{AA'} = \frac{1}{\sqrt{n}}\sum_{x=1}^n  \ket{x}_A \ket{\psi_x}_{A'},
\end{equation}
sends out $A'$ and measures $A$ in the computational basis $\{\ketbra{x}\}_{k=0}^{n-1}$ to obtain $x$. Since the states $\{\ket{\psi_x}\}_x$ span a 2-dimensional space, by the Schmidt theorem the state $\ket{\Phi}_{AA'}\in\mathbb{C}^2\otimes\mathbb{C}^2$ is a two qubit state. 

Without loss of generality an attack performed by Eve starts with an isometry $U$ mapping $A'$ onto systems $B$ and $E$ of arbitrary dimension
\begin{equation}
U: \ket{\Phi}_{AA'} \mapsto \ket{\Psi}_{ABE} = \mathbb{1}_A\otimes U_{A'}\ket{\Phi}_{AA'}.
\end{equation}
In addition Eve chooses a set of binary measurements $\{B_{0|y},B_{1|y}\}$ acting on $B$. 
The combinations of the isometry and the measurements on $B$ define the measurements on the system $A'$ via
\begin{equation}
M_{b|y} = U_{A'}^\dag \left(B_{b|y}\otimes \mathbb{1}_E \right) U_{A'}.
\end{equation}

Furthermore these measurements are constrained to satisfy $\bra{\psi_x} M_{b|y}\ket{\psi_x} = p(b|x,y)$
by the probabilities observed by Alice and Bob in Eq.~\eqref{eq: probas loss only}. From $\bra{\psi_y} M_{0|y}\ket{\psi_y} = 0$, one concludes that $M_{0|y} \propto \id- \ketbra{\psi_x}{\psi_x}$. Any of the remaining probabilities $\bra{\psi_{x\neq y}} M_{0|y}\ket{\psi_{x\neq y}}$ implies
\begin{equation}
M_{0|y} = \eta (\id- \ketbra{\psi_y}{\psi_y})= \eta  \ketbra{\psi_y^\perp}{\psi_y^\perp}.
\end{equation}

This form of $M_{0|y}$ is very restrictive for Eve. In particular, it projects $\ket{\Phi}_{AA'}$ into a product state
\begin{equation}\label{eq: Th2 1}
  \id_A\otimes  \sqrt{M_{0|y}}\ket{\Phi}_{AA'} = \sqrt{p(b=0|y)} \ket{\xi^{(0|y)}}_A \ket{\phi^{(0|y)}}_{A'},
\end{equation}
with $p(b=0|y) = \frac{1}{n}\sum_x p(b=0|x,y)$. This identity can be put in the form
\begin{equation}\label{eq: Th2 2}
\sqrt{B_{0|y}}\otimes \id_{AE} \ket{\Psi}_{ABE} = \sqrt{p(b=0|y)} \ket{\xi^{(0|y)}}_A \ket{\Psi^{(0|y)}}_{BE}.
\end{equation}
From here we can define the marginal state of Alice and Eve conditional to Bob measuring $y$ and obtaining $0$
\begin{equation}\begin{split}
    \rho_{AE|B}^{(0|y)} &= \ketbra{\xi^{(0|y)}}{\xi^{(0|y)}}_A \otimes  \rho_E^{(0|y)} \\
   \textrm{with}\quad &\rho_E^{(0|y)} = \tr_B \ketbra{\Psi^{(0|y)}}{\Psi^{(0|y)}}_{BE}.
\end{split}
\end{equation}
Remarkably, Eve's state is no longer influenced by any manipulations done by Alice, and in particular by her measurement result $x$. That is, conditionally on $y$ Eve's state is independent of $x$. This means that after the sifting Eve can only guess $x$ perfectly ($p_g(e=k|\text{succ})=1$), if she can guess $y$ perfectly. Formally,
\begin{equation}
    p_g(e=k|\text{succ}) = 1 \implies \frac{1}{2} \big\lVert\rho_E^{(0|y)}  -\rho_E^{(0|y')}\big\rVert =0\quad \forall\, y\neq y'.
\end{equation}
Let us now show that this imposes some conditions on the probabilities $p(b=0|y)$. To do so consider the trivial inequality
\begin{equation}
\rho_E  = p(b=0|y) \rho_{E}^{(0|y)} + (1-p(b=0|y)) \rho_{E}^{(1|y)},
\end{equation}
where $\rho_E = \tr_{AB} \ketbra{\Psi}{\Psi}_{ABE}$ and $(1-p(b=0|y)) \rho_{E}^{(1|y)} = \tr_{AB}\big[(B_{1|y}\otimes \id_{AE}) \ketbra{\Psi}{\Psi}_{ABE}\big]$, which implies
\begin{equation}
    \rho_E - p(b=0|y=0) \rho_{E}^{(0|0)}\geq 0.
\end{equation}
But because the state $\rho_{E}^{(0|0)}$ and $\rho_{E}^{(0|1)}$ have orthogonal support, we also obtain
\begin{equation}\label{eq: Th2 3}
    \rho_E - p(b=0|y=0) \rho_{E}^{(0|0)} - p(b=0|y=1) \rho_{E}^{(0|1)}\geq 0.
\end{equation}
By recursion we obtain the bound
\begin{equation}\label{eq: Th2 4}
\begin{split}
      \rho_E - \sum_{y} p(b=0|y)\rho_{E}^{(0|y)} &\geq 0\\
      1 -\sum_y p(b=0|y) &\geq 0 \\
      \sum_y p(b=0|y) &\leq 1
\end{split}
\end{equation}
on the average probability of the $b=0$ outcome. With the help of Eq.~\eqref{eq: probas loss only} this bound can be written as
\begin{equation}
\eta \leq \frac{1}{\frac{1}{n}\sum_{x,y} (1-|G_{xy}|^2)}.
\end{equation}

This bound is, however, worse that the one in the statement of the theorem. Let us now show how to match the two. For this we consider a thought experiment where Alice prepares some pure state
\begin{equation}
\ket*{\widetilde{\Phi}}_{AA'} \in \mathbb{C}^2\otimes\mathbb{C}^2.
\end{equation}
As  $M_{0|y}= \eta \ketbra{\psi_y^{\perp}}{\psi_y^\perp}$ is proportional to  projector on a state, one has, analogously to Eqs.~\eqref{eq: Th2 1}--\eqref{eq: Th2 2},
\begin{equation}
    \left(\sqrt{B_{0|y}} \otimes \id_{AE}\right) U_{A'}\ket{\widetilde{\Phi}}_{AA'} 
\!\! = \sqrt{\widetilde{p}(0|y)}  \ket{\widetilde{\xi}^{(0|y)}}_A \ket{\Psi^{(0|y)}}_{BE},
 \nonumber
\end{equation}
with the same state $\ket{\Psi^{(0|y)}}_{BE}$. Hence the marginal state $\rho_E^{(0|y)}$ are also the same, and satisfy
\begin{equation}
\widetilde{\rho}_E  = \widetilde{p}(b|y) \rho_{E}^{(0|y)} + (1-\widetilde{p}(0|y)) \rho_{E}^{(1|y)}
\end{equation}
for $\widetilde{\rho}_E = \tr_{AB} [U_{A'}\ketbra*{\widetilde{\Phi}}{\widetilde{\Phi}}_{AA'} U_{A'}^\dag]$. We can now repeat the arguments of Eqs.~\eqref{eq: Th2 3}--\eqref{eq: Th2 4} to obtain the bound
\begin{equation}
\sum_y \widetilde{p}(0|y)\leq 1,
\end{equation}
valid for the sum of probabilities
\begin{equation}
\begin{split}
\sum_{y} \widetilde{p}(0|y) &= \tr_{AA'} \Big[\big(\id_A\otimes\sum_y M_{0|y} \big)  \ketbra*{\widetilde{\Phi}}{\widetilde{\Phi}}_{AA'}\Big]\\
     &= \tr_{A'} \Big[\big(\sum_y M_{0|y}\big) \, \rho_{A'}\Big]
\end{split}
\end{equation}
coming from any state $\rho_{A'}$. Choosing the state which  maximizes the bound $\max_{\rho_{A'}} \tr [ \rho_{A'} (\sum_y M_{0|y})] = \lVert\sum_y M_{0|y}\rVert$, one obtains
\begin{equation}\begin{split}
 \Big\lVert\sum_y M_{0|y}\Big\rVert&\leq 1 \\
 \eta \Big\lVert \sum_y (\id -\ketbra{\psi_y}{\psi_y}) \Big\rVert & \leq 1\\
 \eta \big( n - \lambda_\textrm{min}(G)  \big)  & \leq 1 \\
 \eta \leq \frac{1}{ n - \lambda_\textrm{min} (G)},
\end{split}
\end{equation}
where we used the fact that $G$ and $\sum_{y}\ketbra{\psi_y}{\psi_y}$ have the same eigenvalues. Hence, having $p_g(e=k|\text{succ})=1$ and $\eta > \frac{1}{ n - \lambda_\textrm{min} (G)}$ is impossible, which concludes the proof.
\end{proof}

Propositions \ref{prop: ub} and \ref{prop: lb} imply that, for any transmission $\eta$, there exists a RDI-QKD protocol involving $n>\frac{1}{\eta}$ different measurements performed by Bob which yields a positive key rate. In particular, as follows from the proof of Propostion \ref{prop: lb}, this is achieved by the ideal qubit protocol of Sec.~\ref{subseq:concreteex} by choosing $\lambda_{\min}(G)< n -\frac{1}{\eta}$. Conversely, in the RDI setting where Bob can do $n$ different measurements labeled by the settings $y$, Eve can always perform a ``blinding'' attack and obtain a perfect copy of Bob's registers. To do so she performs one of the possible measurements $y'$ at random, records the outcome $e$, and sends a copy of $e$ and $y'$ to Bob's detector. When Bob performs his measurement with a setting $y$, the detector reveals $b=e$ if $y=y'$ and pretends that the system was lost $b=\emptyset$ otherwise. Since $p(y'= y|y)=1/n$, for $\eta<\frac{1}{n}$ Eve is left with a perfect copy of Bob's registers $(b,y)$ whenever the detection is successful $b\neq\emptyset$.

\subsection{Importance of the choice of the Gram matrix} 
\label{sec:Gram_importance}

As we saw in the previous section, if one chooses the $n$ states $\{\ket{\psi_x}\}_x$ well then one can obtain $R>0$, and thus a positive key rate, whenever $\eta > 1/n$.
In this section, we show that it is indeed important to choose the Gram matrix constraining the preparations with some care.
In particular, we show that for a seemingly natural choice of Gram matrix the critical transmission, below which no key can be obtained, is significantly worse: Alice and Bob will not be able to provide a nontrivial lower bound on the key rate if there is more than $50\%$ loss, i.e.\ if $\eta > 1/2$.

We assume thus that Alice prepares a set of $n$ quantum states compatible with the Gram matrix $G_{xx'} = \braket{\psi_x}{\psi_{x'}} = d$ with $d\in(0,1)$ for all $x\neq x'$ and that the observed statistics are given by Eq.~\eqref{eq: probas loss only}. Since $\rank(G) = n$, the states she prepares are necessarily linearly independent. As a result, there exists an unambiguous state discrimination (USD) measurement \cite{C98}. Because of the symmetry, we consider an equiprobable USD and the probability of a conclusive discrimination is given by the smallest eigenvalue of the Gram matrix which is in our case equal to $1-d$ \cite{HEK19}. 

We assume that Eve performs an intercept-resend attack such that with a probability $0\leq q \leq 1$ she performs USD and forces Bob's detection, and with a probability $1-q$ she leaves the state untouched and guesses at random. Given that $x\neq y$, if Eve attacks the USD is conclusive with a probability $1-d$ and if she does not intercept the state Bob gets $b=0$ with a probability $\eta(1-d^2)$. Eve wants her attack to remain unnoticed and this fixes the probability $q$ of intercepting the state to $q = \frac{(1+d)(1-\eta)}{d}$. 

The probability that Eve successfully guesses the secret bit is then given by
\begin{equation}
    \begin{aligned}
        p_g(e=x|\text{succ}) =\frac{1-\eta(1-d)}{2d\eta}.
    \end{aligned}
\end{equation}
Eve thus has entire knowledge of the secret bit string, i.e., $p_g(e=k|\text{succ})=1$, for $\eta=\frac{1}{1+d}> \frac{1}{2}$. Hence, considering identical real overlaps prevents Alice and Bob from obtaining a positive key rate for more than $50\%$ loss.

\section{Bridging the gap between the protocols and practical implementations}
\label{sec: noisy states}

Our receiver-device independent setting assumes the characterization of Alice's state preparation device, given by the Gram matrix
\begin{equation}
\textrm{Alice} \simeq G.
\end{equation}
When Alice prepares pure states, as we have assumed so far, the Gram matrix gives an exhaustive description of the state preparation for our purpose.
That is, in the considered RDI setting, additional information on the states does not help restricting Eve further. In particular, any common unitary transformation or isometry on the states can be cancelled by Eve and does not affect the attacks she can perform. 

In practice the pure state assumption is always an idealization. Here, we discuss how a more realistic model of Alice's setup can be analysed with our protocols.

\subsection{Mixed state models}
\label{subseq:mixedstates}

Here we consider the setting where Alice's preparation device sends out a mixed state $\rho_x$ for each possible value of $x$. That is the preparation box is modeled by a set of mixed states 
\begin{equation}\label{eq: mixed}
\textrm{Alice} \simeq \{\rho_x\}_{x=0}^{n-1}.
\end{equation}

An ensemble of mixed states of a system $A'$ can be jointly purified onto a larger system $A'\otimes A_\text{aux}$ to obtain a set of pure states $\{\ket{\psi_x}\}_{x=0}^{n-1}$ with $\ket{\psi_x} \in \mathcal{H}_{A'}\otimes \mathcal{H}_{A_\text{aux}} $ and 
\begin{equation}\nonumber
\qquad \rho_x = \tr_{A_\text{aux}} \ketbra{\psi_x}{\psi_x}\ \forall \, x.
\end{equation}
Because the system $A_\text{aux}$ remains inside Alice's lab by assumption, any security guarantee obtained for a Gram matrix $G_{\{\ket{\psi_x}\}}$ induced by the set of pure states $\{\ket{\psi_x}\}_{x=0}^{n-1}$ is valid for the original mixed states. In this case one is interested in finding the best-case purification maximizing the key rate. This gives a straightforward way to apply our protocols to noisy preparation devices modeled by Eq.~\eqref{eq: mixed}. The resulting bounds are not necessarily tight, because in the analysis the purifying system $A_\text{aux}$ is given to the eavesdropper, but are secure.

An interesting open question is whether there exists a compressed representation of the mixed state ensemble $\{\rho_x\}_{x=0}^{n-1}$, analogous to the Gram matrix, that specifies all the relations between the states useful for our purpose. Notably, in the case of two states the fidelity between them $F(\rho_0,\rho_1)$ precisely corresponds to the maximal fidelity between their purifications (see e.g.~\cite{nielsen2002quantum}). However, for larger ensembles the knowledge of pairwise fidelities is known to be insufficient to characterize the joint purification~\cite{fannes2011matrices}. As a simple example note that even in the case of three pure states the pairwise fidelities disregard the complex phases of the Gram matrix entries, which can be crucial for the security analysis as we have seen in Section~\ref{sec:Gram_importance}.

\subsection{Fully characterized correlated noise models}

Next, let us consider the general situation where Alice's preparation device is well described by pure states that are however subject to noise, e.g., coming from drifts and fluctuations of some parameters (laser amplitude, phase noise etc). In such a case the preparation box is modelled by a parametric set of states
\begin{equation}\label{eq: parametric proba}
\textrm{Alice} \simeq \{\ket{\phi_x(\lambda)}\}_{x=0}^{n-1},
\end{equation}
where $p(\lambda)$ is the distribution of the noise parameter $\lambda$. In contrast to Eq.~\eqref{eq: mixed}, this model allows for correlated noise affecting the preparation device for all measurement settings. Notably, the model in Eq.~\eqref{eq: parametric proba} reduces to  Eq.~\eqref{eq: mixed} when the hidden variable $\lambda = (\lambda_0,\dots, \lambda_{n-1})$ is composed of random variables $\lambda_x$  that only influence the preparation for the respective setting $x$ and are distributed independently.

Each set of pure states labeled by $\lambda$ corresponds to a Gram matrix $G(\lambda)$. Here, it is important to realize that the correlations $p(b|x,y)$ observed by Alice and Bob do not constrain each $\lambda$ (unless the distribution $p(b|x,y)$ is extremal) but are only respected on average, i.e.
\begin{equation}
p(b|x,y) = \int \dd\lambda\, p(\lambda)\, p(b|x,y,\lambda) 
\end{equation}
for some hidden $p(b|x,y,\lambda)$. Hence, one cannot simply verify the security of the protocol for each $G(\lambda)$.

Instead, we recover a pure state situation by explicitly including the hidden noise parameter $\lambda$ in the state. That is, we consider Alice preparing states of the form 
\begin{equation}\label{eq: states with noise label}
    \ket{\psi_x} = \int \dd \lambda \sqrt{p(\lambda)} \ket{\phi_x(\lambda)}\ket{\lambda},
\end{equation}
with the ``label'' states for the hidden noise parameter respecting  $\braket{\lambda}{\mu}=\delta(\lambda-\mu)$. By doing so we give the noise label $\lambda$ to Eve who can control it coherently but is bound to respect our noisy model of the device given by $p(\lambda)$. It is straightforward to see that the resulting Gram matrix for the states $\{\ket{\psi_x}\}_{x=0}^{n-1}$ is simply the average
\begin{equation}
\overline G = \int \dd\lambda \, p(\lambda) \, G(\lambda).   
\end{equation}
Consequently, verifying the security of the protocol for $\overline G$ guarantees its security for the original model.

\subsection{Partially characterized correlated noise models}

In some situations the full model with the knowledge of the distribution $p(\lambda)$ might not be appropriate, as it requires a complete, precise characterization of the noise mechanisms present. Instead one can only guarantee (with the desired level of confidence) that in each round the preparation device obeys to the model 
\begin{equation}\label{eq: parametric range}
\textrm{Alice} \simeq \{\ket{\phi_x(\lambda)}\}_{x=0}^{n-1} \quad \textrm{with} \quad \lambda \in \Lambda,
\end{equation}
where $\Lambda$ specifies the range of possible $\lambda$. From there we can recover the previous case by noting that any realization of such model corresponds to the states $\{\ket{\psi_x}\}_{x=0}^{n-1}$ in Eq.~\eqref{eq: states with noise label} for \textit{some} probability density $p(\lambda)$ on $\Lambda$. The resulting average Gram matrix then necessarily belongs to the set
\begin{equation}\begin{split}
\overline G \, &\in\,  \widehat{G}_\Lambda \\
G_\Lambda &= \{ G(\lambda)|\lambda \in \Lambda \},
    \end{split}
\end{equation}
where the hat $\widehat{G}_\Lambda$ denotes the convex hull of the set $G_\Lambda$. In principle, it remains to determine the worst case $\overline G$ inside the set, with respect to the key rate it implies, in order to guarantee the security for the noise model. 

Practically, this problem is however computationally hard. And instead of solving it directly it is convenient to further relax the constraints on $\overline G$ to a form that one can easily include in the security analysis described in Section~\ref{sec:securityAnalysis}. This can be done by constraining each entry of the Gram matrix $\overline G_{ij}$ \textit{independently}. Concretely, the set $G_\Lambda$ can be relaxed to a collection of constraints
\begin{equation}
    \begin{split}
        r_{ij}\leq &\, \textrm{Re}[G_{ij}]\,\leq R_{ij}\\
        i_{ij}\leq &\, \textrm{Im}[G_{ij}]\,\leq I_{ij},
    \end{split}
\end{equation}
on the real and imaginary parts of each entry of the matrices $G\in G_\Lambda$. Being linear these constraints remain valid for the convex hull set $\widehat{G}_\Lambda$. Most importantly, they are very simple to include in the SDP. The equality constraint $\Gamma^{11}_{ij} = G_{ij}$ in Eq.~\eqref{eq: sdpconsG} translates in two inequalities on the real and imaginary part of $\Gamma$ for $i\neq j$ 
\begin{equation}
    \begin{split}
        r_{ij}\leq &\, \textrm{Re}[\Gamma^{\id\id}_{ij}] \,\leq R_{ij}\\
        i_{ij}\leq &\, \textrm{Im}[\Gamma^{\id\id}_{ij}] \,\leq I_{ij}.
    \end{split}
\end{equation}
Through the SDP the Gram matrix is constraint to be positive semidefinite. Hence, the set of states described by the Gram matrix which maximizes $p_g(e=k|\text{succ})$ remains physical.

\section{Comparison to other QKD models}
\label{sec:comparison}

In this section we present a brief comparison of our RDI protocols with other models for partially DI QKD. 

Let us start with the one-sided DI model first proposed in Ref.~\cite{Tomamichel2011}. The model applies to a prepare-and-measure scenario, where the receiver is untrusted (as in the RDI model), while the sender uses a fully characterized device. Ref.~\cite{Tomamichel2011} demonstrates security with the help of generalized tripartite uncertainty relations, but no practical considerations are discussed. Subsequent works \cite{Tomamichel2012, Tomamichel2017} discussed the practical requirements of such an approach, considering the effect of noise and finite-size data; see also Refs \cite{Koashi2006,Hayashi2012} for similar analysis based on different proof techniques. However, the effect of losses is either not discussed~\cite{Tomamichel2011,Koashi2006}, or their analysis is based on a fair sampling type assumption \cite{Tomamichel2012,Hayashi2012,Tomamichel2017}, where the detection of the photon is assumed to be independent from the choice of measurement made by the receiver. Note that the fair sampling  assumption allows one to attribute the non-detection events to a filter applied on the system before the measurement and essentially discard the no-click events in the security analysis \cite{OrsucciFS}.
In practice, however, such an assumption is hard to justify in adversarial scenarios like QKD, as it opens the door to blinding attacks exploiting the fact that Eve can steer Bob's detector to click or not depending on the measurement setting~\cite{Lydersen2010}. Therefore, the results of Ref.~\cite{Tomamichel2012, Tomamichel2017} (and also \cite{Koashi2006,Hayashi2012}) cannot be applied to a practical QKD setup without sacrificing the one-sided DI security. Another notable point is that the security analysis of~\cite{Tomamichel2012, Tomamichel2017} relies on an entropic uncertainty relation for a pair of measurements. In the prepare-and-measure setting this approach thus applies to protocols where Alice prepares four states (grouped in two pairs of orthogonal states), and Bob performs two measurements. It is unclear to us whether such an approach can tackle more general protocols \footnote{Note that uncertainty relations for more than two measurements have been investigated, but obtaining good bounds for them appears to be challenging.}, with more states and measurements. Indeed, these cases are important, as the number of measurements of Bob must become large in order to accommodate for low transmissions. Notably, the protocols that we analyze here using the overlap method, where Alice prepares an increasing number of qubit states, can allow for an arbitrarily low transmission.

The one-sided DI model can also be investigated in an entanglement-based scheme where one of the parties is trusted while the other is considered as a black-box \cite{Tomamichel2011}. This approach was discussed in Ref.~\cite{Branciard2012}, establishing a connection with the effect of quantum steering. Here, both noise and losses are taken into account. The requirements in terms of detection efficiency are high ($\eta > 65.9 \%$), hence providing only minor improvements over the DI model. In practice, entanglement-based one-sided DI QKD has never been implemented.

Our RDI protocols therefore provide a number of improvements over previous works on the one-sided DI scenario. First, our protocols are shown to be secure in a prepare-and-measure scenario taking into account both noise and losses; note that the companion paper \cite{Ioannou2021} considers also finite-size effect for the proof-of-principle experiment. In particular, RDI protocols can in principle allow for an arbitrarily low transmission, as we discussed. Compared to the approach of Ref.~\cite{Branciard2012}, the experimental realisation is greatly simplified, as no source of entanglement is necessary, and much lower detection efficiencies can be tolerated. Moreover, in our case, the characterized party (Alice) acts as a sender, while in the one-sided model Alice holds a measurement device. Having to trust a preparation device instead of a measurement device is arguably an advantage.

Another SDI approach presented in Ref.~\cite{Pawlowski2011,WP2015} shares more similarities with our approach. The authors consider prepare-and-measure scenario where Alice's device is assumed to prepare quantum states of bounded Hilbert space dimension (for instance qubits), while Bob's device is completely black-box. This represents a very different type of assumption on the preparations, which is however arguably difficult to justify in practice; indeed a photon is not a qubit, and has many other degrees of freedom than (say) polarisation \footnote{More generally, practical systems are usually described as infinite dimensional systems, for instance when considering optical implementations based on coherent states.}. In this sense, we believe that our RDI approach is more naturally tailored to experiments, as it can deal with systems of arbitrary (possibly infinite) dimension. Another important advantage in practice, is the robustness to losses. Indeed, dimension-based SDI protocols are also sensitive to detection-loophole-type attacks and thus require detection efficiencies comparable to Bell tests \cite{dallarno2015}. This renders their practical implementation challenging. To the best of our knowledge, no experiment has been reported so far. Another related approach was developed in Ref. \cite{Goh2016}, considering an entanglement-based QDK setup assuming only the dimension of the entangled state prepared by the source. Again, practical implementation is challenging due to high detection efficiency requirements.

Finally, we compare our RDI model to the MDI approach \cite{Lo2012,Braunstein12}. Both approaches aim at relaxing trust on the measurement device. While we do this in the prepare-and-measure scenario, the MDI model considers an additional party (Charlie), located in between Alice and Bob and who acts as a relay. Charlie's (measurement) device is then fully untrusted, while Alice's and Bob's (preparation) devices must be well characterized. In practice, a strong advantage of the MDI approach is its robustness to losses, leading to record-distance experiments \cite{Yin2016,Pittaluga2021,Chen2021511,Chen2021658}. In a scenario where both end parties, Alice and Bob, have means to characterize and test their devices (or good reasons to believe the devices function correctly), the MDI approach is a good choice. However, in a scenario where one of the parties does not have the resources (or the expertise) for testing and characterizing their device (or reasons not to trust their devices, for instance a possible malfunctioning due to ageing), the RDI approach provides a good solution. In contrast, the MDI approach cannot be used here, as Bob's (nor Alice's) device can be described by a black-box; some level of trust on both Alice and Bob will always be required in the MDI case.

\section{Conclusion}

We have discussed QKD protocols considering a receiver-device-independent (RDI) model. We presented a security analysis and investigated limitations of these protocols. Notably, we showed that our protocols can in principle allow for an arbitrarily low transmission (detection efficiency). We also provided a detailed discussion concerning the relevance of our approach in a practical context, in particular discussing how the overlap assumption can be justified. These results complement a recent (companion) paper, where a proof-of-principle RDI QKD experiment has been reported~\cite{Ioannou2021}.

To conclude, we discuss a number of open questions. A first interesting question is to derive stronger bounds on the secret key rate. This may be possible using techniques recently developed in Ref.~\cite{Brown2021} providing lower bounds on the conditional von Neumann entropy (instead of the conditional min-entropy, as we consider here) from observed data. Elements from the approach used in Ref.~\cite{Tan2021} might also be useful here.

Another question is to turn our asymptotic key rate into a finite key length when a finite number of systems are exchanges between Alice and Bob. A natural route towards this goal consists in using the entropy accumulation theorem \cite{EAT}, although it is still unclear whether this approach can be adapted to the prepare-and-measure scenario.

An important direction to pursue is to look for RDI protocols that can achieve long distance and are practical. Here we presented protocols that can tolerate the minimum possible transmission (depending on the number of measurements $n$ made by Bob) in the RDI model. In practice, the drawback of our protocols is the sifting, which, for large $n$, renders the protocols inefficient. Developing more efficient protocols would represent significant progress.

Finally, we note our approach shares similarities with the recent work of Ref.~\cite{Tavakoli21}, where the author investigates correlations in a prepare-and-measure scenario with bounded distrust in the preparations. Specifically, the fidelity of the prepared states with respect to some reference state is lower bounded. Hence the distance between the actual and ideal states is bounded. In our approach we bound the distance between the prepared states via their pairwise overlaps.

\medskip

\emph{Acknowledgements.---}We thank Jonathan Brask, Davide Rusca and Hugo Zbinden for discussions. We acknowledge financial support from the EU Quantum Flagship project QRANGE, and the Swiss National Science Foundation (BRIDGE, project 2000021\_192244/1 and NCCR QSIT).

\bibliography{SDIQKD}

\end{document}